\documentclass[aps,prl,10pt,amsmath,twocolumn,floatfix,longbibliography]{revtex4-1}
\usepackage[dvipsnames]{xcolor}
\usepackage[colorlinks=true,bookmarks=false,linkcolor=NavyBlue,urlcolor=NavyBlue,citecolor=NavyBlue,breaklinks]{hyperref}
\usepackage{amsmath}
\usepackage{amsfonts}
\usepackage{tabularx}
\usepackage{graphicx}
\usepackage{times}
\usepackage{mathtools}
\usepackage{braket}
\usepackage{amsthm}
 
\newtheorem{theorem}{Theorem}

\newtheorem{proposition}{Proposition}

\theoremstyle{remark}

\newcommand{\parti}[2]{\frac{\partial #1}{\partial #2}}

\newcommand{\abs}[1]{\left|#1\right|}
\newcommand{\bk}[1]{\left(#1\right)}
\newcommand{\Bk}[1]{\left[#1\right]}
\newcommand{\BK}[1]{\left\{#1\right\}}

\newcommand{\norm}[1]{\lVert #1 \rVert}

\DeclareMathOperator{\trace}{tr}
\DeclareMathOperator{\real}{Re}
\DeclareMathOperator{\imag}{Im}

\begin{document}

\title{The Holevo Cram\'er-Rao bound is at most
thrice the Helstrom version}

\author{Mankei Tsang}
\email{mankei@nus.edu.sg}
\homepage{https://blog.nus.edu.sg/mankei/}
\affiliation{Department of Electrical and Computer Engineering,
  National University of Singapore, 4 Engineering Drive 3, Singapore
  117583}

\affiliation{Department of Physics, National University of Singapore,
  2 Science Drive 3, Singapore 117551}

\date{\today}


\begin{abstract}
  In quantum metrology, the Holevo Cram\'er-Rao bound has attracted
  renewed interest in recent years due to its superiority over the
  Helstrom Cram\'er-Rao bound and its asymptotic attainability for
  multiparameter estimation. Its evaluation, however, is often much
  more difficult than that of the Helstrom version, calling into
  question the actual improvement offered by the Holevo bound and
  whether it is worth the trouble. Here I prove that the Holevo bound
  is at most thrice the Helstrom version, so the improvement must be
  limited and the role of incompatibility in quantum estimation turns
  out to be modest. The result also shows that the Helstrom version
  remains a pretty good bound even for multiple parameters, as it can
  be approached asymptotically to within a factor of 3.

  \textbf{Update}: References~\cite{carollo19,carollo20,tsang20}
  supersede this work by proving that the Holevo Cram\'er-Rao bound
  is, in fact, at most twice the Helstrom version and the factor-of-2
  bound can be tight.
\end{abstract}

\maketitle

For any measurement of a quantum system and any unbiased estimator, a
quantum generalization of the Cram\'er-Rao bound (CRB)---first
proposed by Helstrom in 1967 \cite{helstrom}---can be expressed as
\cite{hayashi05}
\begin{align}
\trace G \Sigma &\ge C^{S} \equiv \min_{X \in \mathcal X} \trace \real Q(X),
\label{helstrom}
\\
Q(X) &\equiv \sqrt{G} Z(X)\sqrt{G},
\quad
Z_{\mu\nu}(X) \equiv \trace \rho X_\mu X_\nu,
\end{align}
where $\Sigma$ is the error covariance matrix, $G$ is a real
 and positive-semidefinite cost matrix, $\rho$ is the
density operator of the quantum system that depends on $n$ real
unknown parameters $\theta = (\theta_1,\theta_2,\dots,\theta_n)$,
$\mathcal X$ is the set of all vectoral Hermitian operators
$X = (X_1,X_2,\dots,X_n)$ that satisfy
$\trace X_\mu\partial \rho/\partial \theta_\nu = \delta_{\mu\nu}$, and
the real part of a matrix is defined by
$(\real Q)_{\mu\nu} = \real(Q_{\mu\nu}) =
[Q_{\mu\nu}+(Q_{\mu\nu})^*]/2$.  The original form of $C^{S}$ in terms
of the symmetric logarithmic derivatives of $\rho$
\cite{helstrom,hayashi05} is a closed-form solution of
Eq.~(\ref{helstrom}). The Helstrom CRB serves as a fundamental limit
to quantum estimation and has found many applications in quantum
metrology \cite{glm2011,demkowicz15,pirandola18,pezze18,tsang19a}.

Despite the popularity of the Helstrom CRB, better bounds exist
\cite{holevo11,personick71,yuen_lax,tsuda05,qzzb,glm2012,hall_prx,qbzzb,qwwb,nair18,rubio19,
  rubio20}. In particular, Holevo proposed a bound that can be
expressed as \cite{holevo11,hayashi05,gill_guta,yamagata13}
\begin{align}
  \trace G \Sigma &\ge C^{H} \ge \max\BK{C^{S},C^{R}},
  \\
  C^{H} &\equiv \min_{X \in \mathcal X} \Bk{\trace\real Q(X)+ \norm{\imag Q(X)}_1},
\label{hcrb}
\end{align}
where $C^{R}$ is another CRB due to Yuen and Lax \cite{yuen_lax} that
is not elaborated here, the imaginary part of a matrix is defined by
$(\imag Q)_{\mu\nu} = \imag(Q_{\mu\nu}) =
[Q_{\mu\nu}-(Q_{\mu\nu})^*]/(2i)$, the trace norm is defined as
$\norm{A}_1 \equiv \trace \sqrt{A^\dagger A}$, and $\dagger$ denotes
the conjugate transpose.  When there are multiple parameters, the
Holevo CRB $C^{H}$ is not only tighter but also attainable
asymptotically \cite{gill_guta,yamagata13,yang19}, as it accounts
properly for any incompatibility of the observables that should be
measured. The bound has attracted renewed interest in recent years
\cite{gill_guta,matsumoto02,hayashi05,yamagata13,ragy16,szczykulska16,bradshaw17,bradshaw18,yang19,albarelli19,suzuki19},
as many applications involve multiple unknown parameters and the
effect of incompatibility is of both fundamental and practical
interest.

Despite the fundamental importance of the Holevo CRB, its evaluation
is difficult and daunting numerics is often needed.  This is in
contrast to the more amenable Helstrom CRB, for which many fruitful
computation techniques have been devised over the years
\cite{helstrom,hayashi05,glm2011,demkowicz15,pirandola18,pezze18,tsang19a,hayashi,paris,escher,twc,guta11,tsang_open,alipour,ng16,yuan17a,sidhu20,genoni19,tsang19c}. For
researchers who are reluctant to undertake the endeavor, this raises
the questions how much improvement the Holevo CRB can actually offer
and whether it is worth the trouble after all.  The following theorem
gives a concrete answer.
\begin{theorem}
$C^{H} \le 3 C^{S}$.
\label{thrice}
\end{theorem}

\begin{proof}
  For any $X$, it can be shown that $\sqrt{G}$, $Z$, $Q$, and
  $\real Q$ are positive-semidefinite, $\sqrt{G}$, $\real Q$, and
  $\imag Q$ are real, $i\imag Q$ is Hermitian, and $\imag Q$ is
  skew-symmetric. With
\begin{align}
Q &= \real Q + i \imag Q,
&
i\imag Q &= Q - \real Q,
\end{align}
one can derive an uncertainty relation given by
\begin{align}
\norm{\imag Q}_1 &= \norm{i\imag Q}_1
= \norm{Q - \real Q}_1
\le \norm{Q}_1 + \norm{\real Q}_1
\nonumber\\
&= \trace Q + \trace \real Q = 2 \trace \real Q,
\label{ineq}
\end{align}
where the triangle inequality is used, $\norm{Q}_1 = \trace Q$ and
$\norm{\real Q}_1 = \trace \real Q$ because $Q$ and $\real Q$ are
positive-semidefinite, and
$\trace Q = \trace \real Q + i \trace \imag Q = \trace \real Q$
because $\imag Q$ is skew-symmetric. Now write the Helstrom CRB as
\begin{align}
C^{S} &= \trace \real Q(X^{S}), 
\label{helstrom2}
\end{align}
where $X^{S}$ is the element in $\mathcal X$ that minimizes
$\trace \real Q(X)$ in Eq.~(\ref{helstrom}).  Combining
Eqs.~(\ref{hcrb}), (\ref{ineq}), and (\ref{helstrom2}), one obtains
\begin{align}
C^{H} &\le \trace\real Q(X^{S}) + \norm{\imag Q(X^{S})}_1
\label{CH_bound}
\\
&\le 3 \trace\real Q(X^{S}) = 3 C^{S}.
\end{align}
\end{proof}

Theorem~\ref{thrice} puts the Holevo CRB in the sandwich
\begin{align}
\max\BK{C^{S},C^{R}} \le C^{H} \le 3C^{S},
\end{align}
and researchers can now decide for themselves whether an improvement
by at most a factor of 3 warrants the extra effort of evaluating
$C^{H}$. The theorem may even be on the generous side, as numerical
results often show that the improvement is less than a factor of 2
\cite{bradshaw17,bradshaw18,albarelli19}. On the flip side,
Theorem~\ref{thrice}, together with the asymptotic attainability of
$C^{H}$ \cite{gill_guta,yamagata13,yang19}, implies that $C^{S}$ is
asymptotically approachable to within a factor of 3, so the Helstrom
CRB turns out to be a pretty good bound after all, even for multiple
parameters.

The bound in Theorem~\ref{thrice} can be further tightened in special
cases. Here I consider the cases where $G$ is rank-one or rank-two.
\begin{proposition}
If $G$ is rank-one, $C^H = C^S$.
\label{rankone}
\end{proposition}

\begin{proof}
 A rank-one $G$ can be expressed as $G = g e e^\top$,
where $g$ is its real and positive eigenvalue,
$e$ is the real unit eigenvector, and $\top$ denotes
the transpose. Then 
\begin{align}
\sqrt{G} &= \sqrt{g} e e^\top,
&
Q &=\sqrt{G} Z \sqrt{G} =  g(e^\top Z e) e e^\top.
\end{align}
Since $Z \ge 0$, $e^\top Z e$ is real and nonnegative, meaning that
$Q$ is real and $\imag Q = 0$. The $C^H$ given by Eq.~(\ref{hcrb}) is
hence equal to the $C^S$ given by Eq.~(\ref{helstrom}).
\end{proof}

Note that a rank-one $G$ is not the same as the case of $n = 1$
unknown parameter. To be specific, let $\beta(\theta)$ be a scalar
parameter of interest that depends on the $n$ unknown parameters
$\theta$.  For example, if $\beta(\theta) = \theta_1$, then the rest
$(\theta_2,\dots,\theta_n)$ are nuisance parameters, which may hamper
the estimation. A bound on the error of estimating $\beta$ in the
presence of the many unknowns can be obtained by assuming
\begin{align}
G_{\mu\nu} &= \parti{\beta}{\theta_\mu}\parti{\beta}{\theta_\nu}.
\end{align}
It follows that
\begin{align}
Q &= \bk{\trace \rho Y^2} e e^\top,
\quad
Y = \sum_{\mu=1}^n \parti{\beta}{\theta_\mu}X_\mu,
\\
C^H &= C^S = \min_{Y \in \mathcal Y} \trace \rho Y^2,
\end{align}
where $\mathcal Y$ is the set of all Hermitian operators that satisfy
the constraints
\begin{align}
\trace Y \parti{\rho}{\theta_\mu} &= \parti{\beta}{\theta_\mu},
\quad \mu = 1,2,\dots,n.
\end{align}
This formulation of $C^S$ is then equivalent to the one in
Ref.~\cite{tsang19c} for semiparametric estimation.  For the problems
studied in Ref.~\cite{tsang19c}, Proposition~\ref{rankone} implies
that the Holevo CRB offers no improvement and the Helstrom CRBs
computed there are asymptotically attainable, at least when $n$ is
finite.

\begin{proposition}
If $G$ is rank-two, $C^H \le 2 C^S$.
\label{ranktwo}
\end{proposition}

\begin{proof}
  Let the positive eigenvalues of a rank-$r$ $G$ be
  $\{g_j:j = 1,2,\dots,r\}$ and the corresponding real unit
  eigenvectors be $\{e^j: j = 1,2,\dots,r\}$. Then
\begin{align}
Q &= \sum_{j=1}^r\sum_{k=1}^r \bk{\trace \rho Y_j Y_k}e^j e^{k\top},
&
Y_j &= \sqrt{g_j} \sum_{\mu=1}^n e^j_{\mu} X_\mu.
\end{align}
If $r = 2$,
\begin{align}
\imag Q &= (\imag \trace \rho Y_1Y_2) \bk{e^1 e^{2\top}- e^2 e^{1\top}},
\\
\norm{\imag Q}_1 &= 2\abs{\imag \trace \rho Y_1Y_2} 
= \abs{i\trace \rho \Bk{Y_1,Y_2}}
\\
&\le \trace \rho Y_1^2 + \trace \rho Y_2^2 = \trace \real Q,
\end{align}
where the inequality comes from
Ref.~\cite[Proposition~2.8.3]{holevo11}.  Hence
\begin{align}
C^{H} &\le \trace\real Q(X^{S}) + \norm{\imag Q(X^{S})}_1
\\
&\le 2 \trace\real Q(X^{S}) = 2 C^{S}.
\end{align}
\end{proof}
Considering only the propositions, one might suspect that $C^H$ could
become significantly higher for a $G$ with a higher rank, but
Theorem~\ref{thrice} settles the general case by imposing a hard limit
for any rank, revealing the surprisingly modest role of
incompatibility in asymptotic quantum estimation.  It remains an open
question whether Theorem~\ref{thrice} can be improved further and a
tighter upper bound on $C^H$ can be found.

Discussions with Francesco Albarelli, Richard Gill, and Madalin Guta
are gratefully acknowledged. This work is supported by the Singapore
National Research Foundation under Project No.~QEP-P7.
\bibliography{research2}

\end{document}